\documentclass[fleqn,12pt,twoside]{article}

\usepackage[headings]{espcrc1}
\readRCS
$Id: espcrc1.tex,v 1.2 2004/02/24 11:22:11 spepping Exp $
\usepackage{graphicx}
\usepackage[figuresright]{rotating}

\newtheorem{theorem}{Theorem}

\newenvironment{proof}[1][Proof.]{\begin{trivlist}
\item[\hskip \labelsep {\bfseries #1}]}{\end{trivlist}}

\newcommand{\AmS}{{\protect\the\textfont2
  A\kern-.1667em\lower.5ex\hbox{M}\kern-.125emS}}

\usepackage{amsmath}
\usepackage{amsfonts}
\title{\begin{center}\Large {\bf A note on upper bounds for the maximum span in interval edge colorings of graphs}\end{center}}

\author{\begin{center}\normalsize R.R. Kamalian$^\dag$$^\ddag$, P.A. Petrosyan$^\dag$$^\S$ \\
\small $^\dag$Institute for Informatics and Automation Problems,\\
National Academy of Sciences, 0014, Armenia,\\
\small $^\ddag$Department of Applied Mathematics and Informatics,\\
Russian-Armenian State University, 0051, Armenia,\\
\small $^\S$Department of Informatics and Applied Mathematics,\\
Yerevan State University, 0025, Armenia,\\
\small e-mail: rrkamalian@yahoo.com,
pet\_petros@ipia.sci.am\end{center}}

\runtitle{A note on upper bounds for the maximum span in interval
edge colorings of graphs}
\runauthor{R.R. Kamalian, P.A. Petrosyan}

\begin{document}

\maketitle

\begin{abstract}
An edge coloring of a graph $G$ with colors $1,2,\ldots ,t$ is
called an interval $t$-coloring if for each $i\in \{1,2,\ldots,t\}$
there is at least one edge of $G$ colored by $i$, the colors of
edges incident to any vertex of $G$ are distinct and form an
interval of integers. In 1994 Asratian and Kamalian proved that if a
connected graph $G$ admits an interval $t$-coloring, then $t\leq
\left( d+1\right) \left( \Delta -1\right) +1$, and if $G$ is also
bipartite, then this upper bound can be improved to $t\leq d\left(
\Delta -1\right) +1$, where $\Delta$ is the maximum degree in $G$
and $d$ is the diameter of $G$. In this paper we show that these
upper bounds can not be significantly improved.


\end{abstract}

\section{Introduction}\

An edge coloring of a graph $G$ with colors $1,2,\ldots ,t$ is
called an interval $t$-coloring if for each $i\in \{1,2,\ldots,t\}$
there is at least one edge of $G$ colored by $i$, the colors of
edges incident to any vertex of $G$ are distinct and form an
interval of integers. The concept of interval edge colorings was
introduced by Asratian and Kamalian \cite{b1}. In \cite{b1} they
proved that if a triangle-free graph $G=\left(V,E\right)$ has an
interval $t$-coloring, then $t\leq \left\vert V\right\vert -1$.
Furthermore, Kamalian \cite{b9} showed that if $G$ admits an
interval $t$-coloring, then $t\leq 2\left\vert V\right\vert -3$.
Giaro, Kubale and Malafiejski \cite{b5} proved that this upper bound
can be improved to $2\left\vert V\right\vert -4$ if $\left\vert
V\right\vert \geq 3$. For a planar graph $G$, Axenovich \cite{b4}
showed that if $G$ has an interval $t$-coloring, then $t\leq
\frac{11}{6}\left\vert V\right\vert$. In \cite{b8,b13} interval edge
colorings of complete graphs, complete bipartite graphs, trees and
$n$-dimensional cubes were investigated. The $NP$-completeness of
the problem of existence of an interval edge coloring of an
arbitrary bipartite graph was shown in \cite{b14}. In papers
\cite{b2,b3,b6,b7,b9,b10,b11} the problem of existence and
construction of interval edge colorings was considered and some
bounds for the number of colors in such colorings of graphs were
given. In particular, it was proved in \cite{b2} that if a connected
graph $G$ admits an interval $t$-coloring, then $t\leq \left(
d+1\right) \left( \Delta -1\right) +1$, and if $G$ is a connected
bipartite, then the upper bound can be improved to $t\leq d\left(
\Delta -1\right) +1$, where $\Delta$ is the maximum degree in $G$
and $d$ is the diameter of $G$. In this paper we show that these
upper bounds can not be significantly improved.\\

\section{Definitions and preliminary results}\

All graphs considered in this paper are finite, undirected and have
no loops or multiple edges. Let $V(G)$ and $E(G)$ denote the sets of
vertices and edges of $G$, respectively. The maximum degree in $G$
is denoted by $\Delta (G)$, the chromatic index of $G$ by $\chi
^{\prime }\left( G\right)$ and the diameter of $G$ by $diam\left(
G\right)$. A partial edge coloring of $G$ is a coloring of some of
the edges of $G$ such that no two adjacent edges receive the same
color. If $\alpha $ is a partial edge coloring of $G$ and $v\in
V(G)$, then $S\left( v,\alpha \right)$ denotes the set of colors of
colored edges incident to $v$.

Let $\left\lfloor a\right\rfloor $ denote the largest integer less
than or equal to $a$. Given two graphs $G_{1}=(V_{1},E_{1})$ and
$G_{2}=(V_{2},E_{2})$, the Cartesian product $G_{1}\square G_{2}$ is
a graph $G =(V,E)$ with the vertex set $V =V_{1}\times V_{2}$ and
the edge set $E =\{((u_{1},u_{2}),(v_{1},v_{2}))|$ $u_{1}=v_{1}$ and
$(u_{2},v_{2})\in E_{2}$ or $u_{2}=v_{2}$ and $(u_{1},v_{1})\in
E_{1}\}$.

The set of all interval colorable graphs is denoted by
$\mathfrak{N}$ \cite{b1,b9}. For a graph $G\in \mathfrak{N}$, the
greatest value of $t$, for which $G$ has an interval $t$-coloring,
is denoted by $W\left(G\right)$.

The terms and concepts that we do not define can be found in
\cite{b15}.\\

We will use the following results.

\begin{theorem}
\label{mytheorem1}\cite{b8}. $~W\left(K_{\Delta,\Delta}\right)=
2\Delta-1$ for any $\Delta\in N$.
\end{theorem}

\begin{theorem}
\label{mytheorem2}\cite{b12,b13}. $~W\left(K_{2^q}\right)\geq
2^{q+1}-2-q$ for any $q\in N$.
\end{theorem}

\begin{theorem}
\label{mytheorem3}\cite{b12}. If $G,H\in \mathfrak{N}$, then
$G\square H\in \mathfrak{N}$.
\end{theorem}

\begin{theorem}
\label{mytheorem4}\cite{b2}. (1) If $G$ is a connected graph and
$G\in \mathfrak{N}$, then
\begin{center}
$W(G)\leq \left(diam(G)+1\right)\left( \Delta(G) -1\right) +1$.
\end{center}
(2) If $G$ is a connected bipartite graph and $G\in
\mathfrak{N}$, then
\begin{center}
$W(G)\leq diam(G)\left(\Delta(G) -1\right) +1$.
\end{center}
\end{theorem}\

\section{Main results}\

\begin{theorem}
\label{mytheorem5} For any integers $d,q\in N$, there is a connected
graph $G$ with $diam(G)=d$,
\begin{center}
$\Delta(G)=\left\{
\begin{tabular}{ll}
$2^{q}-1$, if $d=1$,\\
~~$2^{q}$, ~~~if $d=2$,\\
$2^{q}+1$, if $d\geq 3$,\\
\end{tabular}%
\right.$
\end{center}
such that $G\in \mathfrak{N}$ and
\begin{center}
$W(G)\geq \left\{
\begin{tabular}{ll}
$(d+1)(\Delta(G)-1)-q+2$, if $d=1$,\\
$(d+1)(\Delta(G)-1)-q+1$, if $d=2$,\\
$(d+1)(\Delta(G)-1)-q-2$, if $d\geq 3$.\\
\end{tabular}%
\right.$
\end{center}
\end{theorem}
\begin{proof} For the proof we construct a graph $G_{d,q}$ which
satisfies the condition of the theorem. We define a graph $G_{d,q}$
as follows: $G_{d,q}=P_{d}\square K_{2^q}$. Clearly, $G_{d,q}$ is a
connected graph of diameter $d$ and
\begin{center}
$\Delta(G_{d,q})=\left\{
\begin{tabular}{ll}
$2^{q}-1$, if $d=1$,\\
~~$2^{q}$, ~~~if $d=2$,\\
$2^{q}+1$, if $d\geq 3$.\\
\end{tabular}%
\right.$
\end{center}
Since $P_{d},K_{2^q}\in \mathfrak{N}$, by Theorem \ref{mytheorem3}
we have $G_{d,q}\in \mathfrak{N}$.

Let us show that
\begin{center}
$W(G_{d,q})\geq \left\{
\begin{tabular}{ll}
$(d+1)(\Delta(G_{d,q})-1)-q+2$, if $d=1$,\\
$(d+1)(\Delta(G_{d,q})-1)-q+1$, if $d=2$,\\
$(d+1)(\Delta(G_{d,q})-1)-q-2$, if $d\geq 3$.\\
\end{tabular}%
\right.$
\end{center}

Let $V\left(K_{2^q}\right) =\left\{ v_{1},v_{2},\ldots
,v_{2^q}\right\}$ and

\begin{center}
$V\left(G_{d,q}\right) =\bigcup_{i=1}^{d}V^{i}\left(G_{d,q}\right)$,
\end{center}
\begin{center}
$E\left(G_{d,q}\right)
=\bigcup_{i=1}^{d}E^{i}\left(G_{d,q}\right)\cup
\bigcup_{j=1}^{2^{q}}E_{j}\left(G_{d,q}\right)$
\end{center}
where
\begin{center}
$V^{i}\left(G_{d,q}\right)=\{v_{j}^{(i)}|~1\leq j\leq 2^{q}\}$,
\end{center}
\begin{center}
$E^{i}\left(G_{d,q}\right)=\{(v_{j}^{(i)},v_{k}^{(i)})|~1\leq
j<k\leq 2^{q}\}$,
\end{center}
\begin{center}
$E_{j}\left(G_{d,q}\right)=\{(v_{j}^{(i)},v_{j}^{(i+1)})|~1\leq
i\leq d-1\}$.
\end{center}

For $i=1,2,\ldots,d$ define a subgraph $G_{i}$ of the graph
$G_{d,q}$ in the following way:
\begin{center}
$G_{i}=\left(V^{i}\left(G_{d,q}\right),E^{i}\left(G_{d,q}\right)\right)$.
\end{center}
Clearly, $G_{i}$ is isomorphic to $K_{2^q}$ for $i=1,2,\ldots,d$. By
Theorem \ref{mytheorem2} there exists an interval
$(2^{q+1}-2-q)$-coloring $\alpha$ of the graph $K_{2^q}$.

Define an edge coloring $\beta$ of the subgraphs
$G_{1},G_{2},\ldots,G_{d}$.

For $i=1,2,\ldots,d$ and for every $(v_{j}^{(i)},v_{k}^{(i)})\in
E(G_{i})$ we set:
\begin{center}
$\beta ((v_{j}^{(i)},v_{k}^{(i)}))=\alpha
((v_{j},v_{k}))+(i-1)2^{q}$,
\end{center}
where $1\leq j<k\leq 2^{q}$.

Now we define an edge coloring $\gamma$ of the graph $G_{d,q}$ in
the following way:

for $\forall e\in E(G_{d,q})$
\begin{center}
$\gamma(e)= \left\{
\begin{tabular}{ll}
~~~~~~~~~$\beta(e)$,~~~~~~~~~if $e\in E(G_{i})$, $1\leq i\leq d$,\\
$\max S(v_{j}^{(i)},\beta)+1$, if $e=(v_{j}^{(i)},v_{j}^{(i+1)})\in E_{j}\left(G_{d,q}\right)$, $1\leq i\leq d-1$, $1\leq j\leq 2^{q}$.\\
\end{tabular}%
\right.$
\end{center}

It can be verified that if $d=1$, then $\gamma$ is an interval
$(2^{q+1}-2-q)$-coloring of the graph $G_{1,q}$, if $d=2$, then
$\gamma$ is an interval $(3\cdot2^{q}-2-q)$-coloring of the graph
$G_{2,q}$ and $\gamma$ is an interval $((d+1)2^{q}-2-q)$-coloring of
the graph $G_{d,q}$ for $d\geq 3$. This implies the necessary lower
bounds for $W(G_{d,q})$.~$\square$
\end{proof}

\begin{theorem}
\label{mytheorem6} For any integers $d,\Delta\geq 2$, there is a
connected bipartite graph $G$ with $diam(G)=d$, $\Delta(G)=\Delta$,
such that $G\in \mathfrak{N}$ and $W(G)=d(\Delta-1)+1$.
\end{theorem}
\begin{proof} For the proof we are going to construct a graph
$G_{d,\Delta}$ which satisfies the condition of the theorem. We
consider some cases.

Case 1: $\Delta=2$ or $d=2$.

If $\Delta=2$ and $d\geq 2$, then we take $G_{d,2}=C_{2d}$. Clearly,
$C_{2d}\in \mathfrak{N}$, $diam(C_{2d})=d$ and $\Delta(C_{2d})=2$.
From Theorem \ref{mytheorem4} we have $W(C_{2d})\leq d+1$ for $d\geq
2$.

Now let us show that $W(C_{2d})= d+1$ for $d\geq 2$.

Let $V\left(C_{2d}\right) =\left\{ v_{1},v_{2},\ldots
,v_{2d}\right\}$ and $E\left(C_{2d}\right) =\left\{
(v_{i},v_{i+1})|~1\leq i\leq 2d-1\right\}\cup \{(v_{1},v_{2d})\}$.

Define an edge coloring $\alpha$ of the graph $C_{2d}$ as follows:

1. $\alpha((v_{1},v_{2d}))=1$,

2. $\alpha((v_{i},v_{i+1}))=\alpha((v_{2d-i+1},v_{2d-i}))=i+1$ for
$i=1,2,\ldots,d$.

It is easy to see that $\alpha$ is an interval $(d+1)$-coloring of
the graph $C_{2d}$, thus $W(C_{2d})=d+1$.

If $d=2$ and $\Delta\geq 2$, then we take
$G_{2,\Delta}=K_{\Delta,\Delta}$. Clearly, $K_{\Delta,\Delta}\in
\mathfrak{N}$, $diam(K_{\Delta,\Delta})=2$ and
$\Delta(K_{\Delta,\Delta})=\Delta$. From Theorem \ref{mytheorem1} we
have $W(K_{\Delta,\Delta})=2\Delta-1$.

Case 2: $\Delta,d\geq 3$.

Subcase 2.1: $d$ is even.

Define the graph $G_{d,\Delta}$ as follows:

\begin{center}
$V\left(G_{d,\Delta}\right)=\{u_{j}^{(i)},v_{j}^{(i)}|~1\leq i\leq
\frac{d}{2}, 1\leq j\leq \Delta\}$,
\end{center}
\begin{center}
$E\left(G_{d,\Delta}\right)=E_{1}\cup E_{2}\cup E_{3}\cup E_{4}$,
\end{center}
where
\begin{center}
$E_{1}=\{(u_{i}^{(1)},v_{j}^{(1)})|~1\leq i\leq \Delta, 1\leq j\leq
\Delta\}\setminus \{(u_{\Delta}^{(1)},v_{\Delta}^{(1)})\}$,
\end{center}
\begin{center}
$E_{2}=\bigcup_{i=2}^{\frac{d}{2}-1}\{(u_{j}^{(i)},v_{k}^{(i)})|~1\leq
j\leq \Delta, 1\leq k\leq \Delta\}\setminus
\{(u_{1}^{(i)},v_{1}^{(i)}),(u_{\Delta}^{(i)},v_{\Delta}^{(i)})\}$,
\end{center}
\begin{center}
$E_{3}=\{(u_{\Delta}^{(i-1)},v_{1}^{(i)}),(v_{\Delta}^{(i-1)},u_{1}^{(i)})|~2\leq
i\leq \frac{d}{2}\}$,
\end{center}
\begin{center}
$E_{4}=\{(u_{i}^{(\frac{d}{2})},v_{j}^{(\frac{d}{2})})|~1\leq i\leq
\Delta, 1\leq j\leq \Delta\}\setminus
\{(u_{1}^{(\frac{d}{2})},v_{1}^{(\frac{d}{2})})\}$.
\end{center}
Clearly, $G_{d,\Delta}$ is a connected $\Delta$-regular bipartite
graph of diameter $d$. It is easy to see that $G_{d,\Delta}\in
\mathfrak{N}$. By Theorem \ref{mytheorem4} we have
$W(G_{d,\Delta})\leq d(\Delta-1)+1$.

Now we show that $W(G_{d,\Delta})= d(\Delta-1)+1$.

Define an edge coloring $\beta$ of the graph $G_{d,\Delta}$ in the
following way:\\

1. for every $(u_{j}^{(i)},v_{k}^{(i)})\in E(G_{d,\Delta})$
\begin{center}
$\beta ((u_{j}^{(i)},v_{k}^{(i)}))=(i-1)(2\Delta-1)+j+k-i$,
\end{center}
where $1\leq i\leq \frac{d}{2}$, $1\leq j\leq \Delta$, $1\leq k\leq \Delta$;\\

2. for $i=2,\ldots,\frac{d}{2}$
\begin{center}
$\beta ((u_{\Delta}^{(i-1)},v_{1}^{(i)}))=\beta
((v_{\Delta}^{(i-1)},u_{1}^{(i)}))=(i-1)(2\Delta-1)-i+2$.
\end{center}

It is not difficult to check that $\beta$ is an interval
$(d(\Delta-1)+1)$-coloring of the graph $G_{d,\Delta}$.

Subcase 2.2: $d$ is odd.

If $d=3$, then define the graph $G_{3,\Delta}$ as follows:

\begin{center}
$V\left(G_{3,\Delta}\right)=\{u_{i},u_{i}^{\prime}|~1\leq i\leq
\Delta-1\}\cup \{v_{j},v_{j}^{\prime}|~1\leq j\leq \Delta \}$,
\end{center}

\begin{center}
$E\left(G_{3,\Delta}\right)=\{(u_{i},v_{j}),(u_{i}^{\prime},v_{j}^{\prime})|~1\leq
i\leq \Delta-1, 1\leq j\leq \Delta\}\cup
\{(v_{i},v_{i}^{\prime})|~1\leq i\leq \Delta \}$.
\end{center}

Clearly, $G_{3,\Delta}$ is a connected $\Delta$-regular bipartite
graph of diameter $3$. It is easy to see that $G_{3,\Delta}\in
\mathfrak{N}$. By Theorem \ref{mytheorem4} we have
$W(G_{3,\Delta})\leq 3\Delta-2$.

Now we show that $W(G_{3,\Delta})= 3\Delta-2$.

Define an edge coloring $\gamma$ of the graph $G_{3,\Delta}$ in the
following way:\\

1. for $i=1,2,\ldots,\Delta-1$, $j=1,2,\ldots,\Delta$
\begin{center}
$\gamma ((u_{i},v_{j}))=i+j-1$,
\end{center}

2. for $i=1,2,\ldots,\Delta-1$, $j=1,2,\ldots,\Delta$
\begin{center}
$\gamma ((u_{i}^{\prime},v_{j}^{\prime}))=\Delta +i+j-1$,
\end{center}

3. for $i=1,2,\ldots,\Delta$
\begin{center}
$\gamma ((v_{i},v_{i}^{\prime}))=\Delta +i-1$.
\end{center}

It is not difficult to check that $\gamma$ is an interval
$(3\Delta-2)$-coloring of the graph $G_{3,\Delta}$.

Assume that $d\geq 5$.

Define the graph $G_{d,\Delta}$ as follows:

\begin{center}
$V\left(G_{d,\Delta}\right)=\{a,b_{1},b_{2},\dots,b_{\Delta-3},c,d_{1},d_{2},\dots,d_{\Delta-3}
\}\cup\{u_{j}^{(i)},v_{j}^{(i)}|~1\leq i\leq
\lfloor\frac{d}{2}\rfloor, 1\leq j\leq \Delta\}$,
\end{center}
\begin{center}
$E\left(G_{d,\Delta}\right)=E_{1}\cup E_{2}\cup E_{3}\cup E_{4}\cup
E_{5}$,
\end{center}
where
\begin{center}
$E_{1}=\{(u_{i}^{(1)},v_{j}^{(1)})|~1\leq i\leq \Delta, 1\leq j\leq
\Delta\}\setminus \{(u_{\Delta}^{(1)},v_{\Delta}^{(1)})\}$,
\end{center}
\begin{center}
$E_{2}=\{(u_{\Delta}^{(1)},a),(a,u_{1}^{(2)}),(v_{\Delta}^{(1)},c),(c,v_{1}^{(2)}),(a,c)\}\cup
\{(a,b_{i}),(c,d_{i})|~1\leq i\leq \Delta-3\}$,
\end{center}
\begin{center}
$E_{3}=\bigcup_{i=2}^{\lfloor\frac{d}{2}\rfloor-1}\{(u_{j}^{(i)},v_{k}^{(i)})|~1\leq
j\leq \Delta, 1\leq k\leq \Delta\}\setminus
\{(u_{1}^{(i)},v_{1}^{(i)}),(u_{\Delta}^{(i)},v_{\Delta}^{(i)})\}$,
\end{center}
\begin{center}
$E_{4}=\{(u_{\Delta}^{(i-1)},v_{1}^{(i)}),(v_{\Delta}^{(i-1)},u_{1}^{(i)})|~3\leq
i\leq \lfloor\frac{d}{2}\rfloor\}$,
\end{center}
\begin{center}
$E_{5}=\{(u_{i}^{(\lfloor\frac{d}{2}\rfloor)},v_{j}^{(\lfloor\frac{d}{2}\rfloor)})|~1\leq
i\leq \Delta, 1\leq j\leq \Delta\}\setminus
\{(u_{1}^{(\lfloor\frac{d}{2}\rfloor)},v_{1}^{(\lfloor\frac{d}{2}\rfloor)})\}$.
\end{center}
Clearly, $G_{d,\Delta}$ is a connected  bipartite graph with a
maximum degree $\Delta$ and a diameter $d$.

Now we show that $G_{d,\Delta}\in \mathfrak{N}$ and
$W(G_{d,\Delta})= d(\Delta-1)+1$.

Define an edge coloring $\lambda$ of the graph $G_{d,\Delta}$ in the
following way:\\

1. for every $(u_{j}^{(1)},v_{k}^{(1)})\in E(G_{d,\Delta})$
\begin{center}
$\lambda ((u_{j}^{(1)},v_{k}^{(1)}))=j+k-1$,
\end{center}
where $1\leq j\leq \Delta$, $1\leq k\leq \Delta$;\\

2.\begin{center}
$\lambda((u_{\Delta}^{(1)},a))=\lambda((v_{\Delta}^{(1)},c))=2\Delta-1,
\lambda((a,c))=2\Delta$;
\end{center}

3.\begin{center} $\lambda((a,b_{i}))=\lambda((c,d_{i}))=2\Delta+i$,
$i=1,2,\ldots,\Delta-3$;
\end{center}

4.\begin{center}
$\lambda((a,u_{1}^{(2)}))=\lambda((c,v_{1}^{(2)}))=3\Delta-2$;
\end{center}

5. for every $(u_{j}^{(i)},v_{k}^{(i)})\in E(G_{d,\Delta})$
\begin{center}
$\lambda
((u_{j}^{(i)},v_{k}^{(i)}))=(i-1)(2\Delta-1)+j+k-i+\Delta-1$,
\end{center}
where $2\leq i\leq \lfloor\frac{d}{2}\rfloor$, $1\leq j\leq \Delta$, $1\leq k\leq \Delta$;\\

6. for $i=3,\ldots,\lfloor\frac{d}{2}\rfloor$
\begin{center}
$\lambda((u_{\Delta}^{(i-1)},v_{1}^{(i)}))=\lambda
((v_{\Delta}^{(i-1)},u_{1}^{(i)}))=(i-1)(2\Delta-1)-i+\Delta+1$.
\end{center}

It can be verified that $\lambda$ is an interval
$(d(\Delta-1)+1)$-coloring of the graph $G_{d,\Delta}$, thus
$G_{d,\Delta}\in \mathfrak{N}$ and $W(G_{d,\Delta})\geq
d(\Delta-1)+1$. On the other hand, from Theorem \ref{mytheorem4} we
have $W(G_{d,\Delta})\leq d(\Delta-1)+1$, hence $W(G_{d,\Delta})=
d(\Delta-1)+1$.~$\square$
\end{proof}\


\begin{thebibliography}{99}

\bibitem{b1} A.S. Asratian, R.R. Kamalian, Interval colorings of edges of a
multigraph, Appl. Math. 5 (1987) 25-34 (in Russian).

\bibitem{b2} A.S. Asratian, R.R. Kamalian, Investigation on interval
edge-colorings of graphs, J. Combin. Theory Ser. B 62 (1994) 34-43.

\bibitem{b3} A.S. Asratian, T.M.J. Denley, R. Haggkvist, Bipartite Graphs and
their Applications, Cambridge University Press, Cambridge, 1998.

\bibitem{b4} M.A. Axenovich, On interval colorings of planar graphs, Congr.
Numer. 159 (2002) 77-94.

\bibitem{b5} K. Giaro, M. Kubale, M. Malafiejski, Consecutive colorings of
the edges of general graphs, Discrete Math. 236 (2001) 131-143.

\bibitem{b6} H.M. Hansen, Scheduling with minimum waiting periods, Master's
Thesis, Odense University, Odense, Denmark, 1992 (in Danish).

\bibitem{b7} D. Hanson, C.O.M. Loten, B. Toft, On interval colorings of
bi-regular bipartite graphs, Ars Combin. 50 (1998) 23-32.

\bibitem{b8} R.R. Kamalian, Interval colorings of complete bipartite graphs
and trees, preprint, Comp. Cen. of Acad. Sci. of Armenian SSR,
Erevan, 1989 (in Russian).

\bibitem{b9} R.R. Kamalian, Interval edge colorings of graphs, Doctoral
Thesis, Novosibirsk, 1990.

\bibitem{b10} R.R.Kamalian, A.N.Mirumian, Interval edge colorings of
bipartite graphs of some class, Dokl. NAN RA, 97 (1997) 3-5 (in
Russian).

\bibitem{b11} R.R.Kamalian, P.A. Petrosyan, Interval colorings of some
regular graphs, Math. probl. of comp. sci. 25 (2006) 53-56.

\bibitem{b12} M. Kubale, Graph Colorings, American Mathematical Society, 2004.

\bibitem{b13} P.A. Petrosyan, Interval edge colorings of complete graphs and
$n$-dimensional cubes, Discrete Mathematics, 2009, under review.

\bibitem{b14} S.V. Sevast'janov, Interval colorability of the edges of a
bipartite graph, Metody Diskret. Analiza 50 (1990) 61-72 (in
Russian).

\bibitem{b15} D.B. West, Introduction to Graph Theory, Prentice-Hall, New
Jersey, 1996.

\end{thebibliography}
\end{document}